\documentclass[11pt,a4paper]{article}
\usepackage{amssymb,amsmath,amsthm}
\usepackage{multirow}
\usepackage{blindtext}
\usepackage{amsfonts}
\usepackage{amsmath,amssymb,pxfonts, xcolor}
\usepackage{graphicx}
\usepackage{ulem}
\usepackage{cancel}
\usepackage{tabularx}
\usepackage{tikz}
\usepackage{graphicx}
\usepackage{pgfplots}
\usepackage{xcolor}
\usepackage{soul}
\usepackage{float}
\usepackage{comment}
\usepackage{adjustbox}
\usepackage{booktabs}
\usepackage{arydshln}

\newtheorem{theorem}{Theorem}

\newtheorem{corollary}[theorem]{Corollary}

\newtheorem{proposition}[theorem]{Proposition}
\newtheorem{remark}[theorem]{Remark}

\numberwithin{equation}{section}

\makeatletter
\newcommand{\biggg}[1]{{\hbox{$\left#1\vbox to 20.5pt{}\right.\n@space$}}}
\newcommand{\Biggg}[1]{{\hbox{$\left#1\vbox to 23.5pt{}\right.\n@space$}}}
\newcommand{\bigggg}[1]{{\hbox{$\left#1\vbox to 26.5pt{}\right.\n@space$}}}
\newcommand{\Bigggg}[1]{{\hbox{$\left#1\vbox to 29.5pt{}\right.\n@space$}}}
\newcommand{\biggggg}[1]{{\hbox{$\left#1\vbox to 32.5pt{}\right.\n@space$}}}
\newcommand{\Biggggg}[1]{{\hbox{$\left#1\vbox to 35.5pt{}\right.\n@space$}}}
\newcommand{\bigggggg}[1]{{\hbox{$\left#1\vbox to 38.5pt{}\right.\n@space$}}}
\newcommand{\Bigggggg}[1]{{\hbox{$\left#1\vbox to 41.5pt{}\right.\n@space$}}}

\title{Estimating the Hurst parameter from the zero vanna implied volatility and its dual}
\author{Elisa Al\`{o}s\thanks{Department of Economics and Business, University Pompeu Fabra, and Barcelona GSE. Supported by grant MEC MTM 2016-76420-P} \quad Frido Rolloos \quad Kenichiro Shiraya\thanks{Graduate School of Economics, The University of Tokyo. Supported by CARF.}}

\begin{document}
\maketitle

\begin{abstract}
The covariance between the return of an asset and its realized volatility can be approximated as the difference between two specific implied volatilities. In this paper it is proved that in the small time-to-maturity limit the approximation error tends to zero. In addition a direct relation between the short time-to-maturity covariance and slope of the at-the-money implied volatility is established. The limit theorems are valid for stochastic volatility models with Hurst parameter $H \in (0,1)$. An application of the results is to accurately approximate the Hurst parameter using only a discrete set of implied volatilities. Numerical examples under the rough Bergomi model are presented. 

\bigskip

Keywords: Malliavin calculus, fractional volatility models, implied volatility, skew, covariance.

AMS subject classification: 91G99
\end{abstract}

\section{Introduction}

In Rolloos \cite{R1} an approximation is derived that relates the covariance between price returns and realized volatility to the difference between two specific implied volatilities:
\begin{equation}\label{rolloos}
E_t \left[ \left( \frac{S_T}{S_t} - 1 \right) \sqrt{ \frac{1}{T-t} \int_t^T \sigma_u^2 du } \, \right] \approx  I(k_{t,T}^+) - I(k_{t,T}^-) \approx I^2(k_{t,T}^*)(T-t) \frac{ \partial I(k_{t,T}^*)}{ \partial k}
\end{equation}
Here $S_t$ denotes the asset price, $\sigma_t$ its instantaneous volatility, $I(k_{t,T}^-) $ the implied volatility (IV) corresponding to the (log)strike where the Black-Scholes-Merton (BS) vanna and volga of a vanilla option with maturity date $T$ is zero, $I(k_{t,T}^+)$ the IV corresponding to the strike where the BS volga is zero but the BS vanna is nonzero, and $I(k_{t,T}^*)$ is the at-the-money forward (ATM) IV. Notice that the strikes $k_{t,T}^{\pm}, k_{t,T}^*$ are floating strikes, not fixed strikes.

The approximations \eqref{rolloos} are of use to the practitioner for several reasons. For example, in contrast to other more ad-hoc measures of skew such as the difference between the IVs of 90\% and 110\% moneyness, the difference between $I(k_{t,T}^+)$ and $I(k_{t,T}^-)$ has a clear link to covariance. This means that by observing the difference between the two implied volatilities $I(k_{t,T}^{\pm})$ the practitioner can obtain an estimate for the implied covariance between price return and realized volatility. They can then compare the implied measure to subsequent realized covariance and infer whether on average the skew contains a premium. Furthermore, since for small times to maturity the covariance is directly proportional to the ATM slope the approximation complements other results that relates the ATM skew to statistical measures. One such result was obtained by Backus et al. \cite{BFW} which states that for short times to maturity the ATM skew is approximately the skewness of the distribution of the asset log returns.  

As we shall demonstrate the approximations \eqref{rolloos} also enable accurate approximation of the Hurst parameter $H$ using only a discrete set of IVs. To see this recall that as has been shown by Al\`os et al. \cite{ALV} and Fukasawa \cite{F1,F2} the short time-to-maturity $(T-t \ll 1)$ ATM skew has the following property:
\begin{equation}\label{ALVF}
\frac{ \partial I(k_{t,T}^*)}{ \partial k} \propto (T-t)^{H - \frac12}.
\end{equation}
It follows from the above and \eqref{rolloos} that
\begin{equation}\label{Happrox}
H \approx -\frac12 + \frac{ \ln \left( \frac{I(k_{t,T_1}^+) - I(k_{t,T_1}^-) }{I(k_{t,T_2}^+) - I(k_{t,T_2}^-)}\frac{I^2(k_{t,T_2}^*)}{I^2(k_{t,T_1}^*)} \right) }{ \ln \left( \frac{T_1 - t}{T_2-t} \right) }
\end{equation}

This approximation for the Hurst parameter is appealing because it does not depend on specific model parameters. Hence, it can be used without first calibrating a specific stochastic volatility model. If volatility is indeed driven by fractional noise, as put forward by among others Comte and Renault \cite{CR} for $H > 1/2$ and Al\`os, Le\'on and Vives \cite{ALV} for $H < 1/2$, then the approximation \eqref{Happrox} provides a straightforward way to estimate the Hurst parameter from the IV surface. Note also that our estimation of $H$ is easier to carry out than using short-dated volatility swaps (Al\`os and Shiraya \cite{AS}) since volatility swaps are illiquid.

In spite of the practical usefulness of the approximations, Rolloos does not give any insight into the approximation errors. This is due to the fact that the approximations were derived by Taylor expansions of the `mixing' formula due to Hull and White \cite{HW}, Romano and Touzi \cite{RoT}, and Willard \cite{W}. Even though a Taylor expansion is straightforward, in this particular case it is unable to give information on the approximation errors. In contrast, techniques from Malliavin calculus have been shown to be quite powerful in proving limit theorems and providing error bounds. This is the main contribution of our paper: by using techniques from Malliavin calculus we will prove limit theorems for the approximations \eqref{rolloos}.

The paper is structured as follows. In Section 2 the main assumptions are stated and notation introduced. This is followed by Section 3 in which the main limit theorems concerning the difference between an arbitrary IV and the volatility swap strike, and the difference between $I(k_{t,T}^+)$ and $I(k_{t,T}^-)$ are proved. The approximations stated above then follow from the main theorems. In Section 4 we present numerical examples where the rough Bergomi model is taken as benchmark skew generating model. The accuracy of approximations \eqref{rolloos} and \eqref{Happrox} are then examined. The final section concludes.

\section{The main problem and notations}

This paper assumes a stochastic volatility model for a log asset price $X_t := \log S_t$ under the risk-neutral probability measure $P$.
The dynamics is described by the following stochastic differential equation:
\begin{equation}
\label{themodel}
X_{t}=X_0-\frac{1}{2}\int_{0}^{t}{\sigma _{s}^{2}}ds+\int_{0}^{t}\sigma
_{s}\left( \rho dW_{s}+\sqrt{1-\rho ^{2}}dB_{s}\right) ,\quad t\in \lbrack
0,T].
\end{equation}
In this formulation $X_0$ represents the initial log asset price and the dynamics are driven by two independent standard Brownian motions, $W$ and $B$, on a complete probability space $(\Omega ,\mathcal{G},P)$. 
The filtrations generated by $W$ and $B$ are denoted by $\mathcal{F}^{W}$ and $\mathcal{F}^{B}$, respectively, with $\mathcal{F}$ representing their joint filtration, $\mathcal{F}:=\mathcal{F}^ {W}\vee \mathcal{F}^{B}$. 
The volatility component, $\sigma$, is a square-integrable and right-continuous stochastic process that is adapted to $\mathcal{F}^{W}$. 
For the sake of analytical simplicity, we assume a zero interest rate ($r=0$). However, the core arguments of this study remain valid even when $r \neq 0$.

Under this framework the value of a European call option with strike price $K$ is determined by its conditional expectation, as expressed by the following equality:
\[
V_{t}=E_{t}[(e^{X_{T}}-K)_{+}],
\]%
where $E_{t}$ signifies the $\mathcal{F}_{t}$-conditional expectation with respect to measure $P$.

The quantity $v_t=\sqrt{\frac{1}{T-t}\int_{t}^{T}\sigma _{u}^{2}du}$ s the future average volatility over the time to maturity of the option, and this process is not adapted. 
Its conditional expectation, $E_{t}\left[v_{t}\right]$, is known as the fair strike of a volatility swap with maturity $T$.

The price of a European call option in the Black-Scholes model is given by the function $BS(t,T,x,k,\sigma )$, assuming constant volatility $\sigma$, an initial asset price $e^x$, a time to maturity $T-t$, and a strike price $K=\exp(k) $. 
For the special case of zero interest rate we recall that the formula is:
\[
BS(t,T,x,k,\sigma )=e^{x}N(d_1(k,\sigma ))-e^{k}N(d_2(k,\sigma )),
\]%
Here, $N$ is the cumulative distribution function for a standard normal variable, with the parameters $d_1$ and $d_2$ defined as:
\[
d_1\left( k,\sigma \right) :=\frac{x-k}{\sigma \sqrt{T-t}}+ \frac{\sigma }{2}\sqrt{T-t}, \hspace{0.4cm} 
d_2\left( k,\sigma \right) :=\frac{x-k}{\sigma \sqrt{T-t}}- \frac{\sigma }{2}\sqrt{T-t}.
\]%
For convenience we occasionally simplify the notation by writing 
\[
BS(k,\sigma):=BS(t,T,X_t,k,\sigma ).
\]

We also define the inverse Black-Scholes function, $BS^{-1}(t,T,x,k, \cdot)$, with respect to the volatility parameter such that 
\[
BS(t,T,x,k, BS^{-1}( t,T,x,k,\lambda) )=\lambda,
\]
for all $\lambda>0$. 
On occasion we use the following compact notation:
\[
BS^{-1}(k,\lambda)\ :=BS^{-1}(t,T,X_{t},k,\lambda).
\]
 
The implied volatility $I(t,T,X_{t},k)$ is defined for any fixed set of parameters $t,T,X_{t},k$ as the unique volatility value that equates the Black-Scholes price to the market price:
\[
BS( t,T,X_{t},k,I( t,T,X_{t},k) ) =V_{t}.
\] 
It follows directly that 
\[
I(t,T,X_t,k)=BS^{-1}( t,T,X_t,k,V_t).
\]

We define $k_t^* := X_t$ is the ATM strike at time $t$ and $I(t,T,X_t,k^*_t)$ the ATM IV. When interest rate is nonzero then the ATM strike depends on time to maturity $T$ as well. However, under our assumptions there is no dependence on $T$.

The zero vanna strike at time $t$ for maturity date $T$ is the strike $k^-_{t,T}$ for which the Black-Scholes vanna vanishes, i.e.,
\[
d_2(k^-_{t,T},I(t,T,X_t,k^-_{t,T}))=0.
\] 
The Black-Scholes vanna is defined as the first-order partial derivative of the option's delta with respect to the implied volatility and is proportional to $d_2$. 
The implied volatility $I(t,T,X_t,k^-_{t,T})$ is referred to as the zero vanna implied volatility.  A similar definition applies to $I(t,T,X_t,k^+_{t,T})$ and $k^+_{t,T}$, where $d_1(k^+_{t,T},I(t,T,X_t,k^+_{t,T}))=0$. The strike $k^+_{t,T}$ is called the dual zero vanna strike and the corresponding implied volatility $I(t,T,X_t,k^+_{t,T})$ is the dual zero vanna implied volatility.

Even when interest rate and dividend yield are zero the zero vanna strike and its dual depend on time to maturity $T$. For the sake of notational economy, when not considering multiple maturities we shall henceforth drop the dependence of these strikes on $T$ and only reintroduce it in the section on numerical results.

We introduce also the following notations:
\[
\Lambda^-_{r}:=E_{r}\left[ BS( t,T,X_t,k^-_{t,T},v_{t}) \right], \hspace{0.4cm} 
\Lambda^+_{r}:=E_{r}\left[ BS( t,T,X_t,k^+_{t,T},v_{t}) \right].
\]
We also define $\Theta_{r}(k):=BS^{-1}(k,\Lambda_r)$, and note that 
\[
\Theta_{t}(k_{t,T})=I(t,X_t,k_{t,T},V_t), \hspace{0.4cm} 
\Theta_{T}(k^{\pm}_{t,T})=v_t.
\]
Finally, we define the following operators on the Black-Scholes function:
\[
G(t,T,x,k,\sigma ):= \left( \frac{\partial ^{2}}{\partial x^{2}}-\frac{\partial}{\partial x} \right) BS(t,T,x,k,\sigma ), \hspace{0.4cm} H(t,T,x,k,\sigma ):= \left( \frac{\partial ^{3}}{\partial x^{3}}-\frac{\partial ^{2}}{\partial x^{2}} \right) BS(t,T,x,k,\sigma ).
\]

The methodology for the remainder of this paper is grounded in Malliavin calculus. The domain of the Malliavin derivative operator $D^{W}$ with respect to the Brownian motion $W$ is denoted by $\mathbb{D}_{W}^{1,2}$. For $n>1$, the domains for the iterated derivatives $D^{n,W}$ is denoted by $\mathbb{D}_{W}^{n,2}$. We also utilize the notation $\mathbb{L}_{W}^{n,2}=L^{2}\left(\left[ 0,T\right] ;\mathbb{D}_{W}^{n,2}\right)$.

\section{Limit theorems}\label{sec3}

In this section the main assumptions and theorems are stated. Due to the length of the proofs of the theorems they have been placed in the appendix. The proofs of the corollaries are sufficiently short to be kept in the main text. Let us then first state our hypotheses:

\begin{description}
\item[(H1)] There exist two positive constants $a,b$ such that $a\leq \sigma
_{t}\leq b,$ for all $t\in \left[ 0,T\right] .$

\item[(H2)]  $\sigma\in \mathbb{L}^{3,2}_W$ and there exists two constants $\nu>0$ and $H\in (0,1)$ such that, for all $0<r<u,s,\theta<T$
$$
|E_r[D_r^W\sigma_s^2]|\leq \nu(s-r)^{H-\frac12}, \hspace{0.3cm} |E_r[D_\theta^WD_r^W\sigma_s^2]|\leq \nu^2(s-r)^{H-\frac12}(s-\theta)^{H-\frac12},
$$
and 
$$
|E_r[D_u^W D_\theta^W D_r^W\sigma_s^2]|\leq \nu^3(s-r)^{H-\frac12}(s-\theta)^{H-\frac12}(s-u)^{H-\frac12}.
$$
\item[(H3)] Hypotheses (H1), (H2'), hold and the terms 
\[
\frac{1}{(T-t)^{3+2H}}E_{t} \left[\left(\int_{t}^{T}
\int_{s}^{T}D_{s}^{W}\sigma _{r}^{2}dr ds\right) ^{2}\right] ,
\]%
\[
\frac{1}{(T-t)^{2+2H }}E_{t}\left[
\int_{t}^{T} \int_{s}^{T}D_{s}^{W}\sigma_{r} \int_{r}^{T}D_{r}^{W}\sigma _{u}^{2}du dr ds
\right],
\]%
\[
\frac{1}{(T-t)^{2+2H }}E_{t}\left[
\int_{t}^{T} \int_{s}^{T} \int_{r}^{T}D_{s}^{W}D_{r}^{W}\sigma_{u}^{2}du dr ds\right]
\]%
have a finite limit as $T\rightarrow t.$

%{\color{red}\item[(H6)] There exists a positive constant $c$ such that $\left| D_r^W \sigma_s \right| \le c$, for all $0 \le r \le s \le T$.}
\end{description}

The following result follows from the same argument as Proposition 4.1 in Al\`os and Shiraya (2019) and serves as the main tool in this section.
\begin{proposition}\label{Theoremcorrelatedcase}
Consider the model \eqref{themodel} and assume
that hypotheses (H1), (H2) and hold for some $H\in (0,1)$. Then, 
\begin{eqnarray}
I\left( t,T,X_t,{k}\right) &=&I^{0}\left( t,T,X_t,{k}\right)
\nonumber \\
&&+\frac{\rho }{2}E_t \left[\int_{t}^{T}( BS^{-1}) ^{\prime }( k,\Gamma _{s}) H(s,T,X_{s},{k},v_{s})\Phi _{s}ds\right]
\label{impliedrel}
\end{eqnarray}%
where $I^{0}( t,T,X_t,k) $ denotes the implied
volatility in the uncorrelated case $\rho =0$,
\[
\Gamma _{s}:=E_{t}[ BS(t,T,X_t,k,v_{t})] +\frac{\rho }{2}%
E_{t}\left[\int_{t}^{s} H(r,T,X_{r},k,v_{r})\Phi _{r}dr\right], 
\]%
and $\Phi _{t}:=\sigma _{t}\int_{t}^{T}D_{t}^{W}\sigma _{r}^{2}dr.$ 
\end{proposition}

By making use of Proposition \ref{Theoremcorrelatedcase} we can prove the following.

\begin{proposition}\label{newdecomposition}
Consider the model \eqref{themodel}\ and assume
that hypotheses (H1), (H2') and (H3)  hold for some $H\in(0,1) $. Then
\begin{eqnarray*}
\lefteqn{I\left( t,T,X_{t},k\right) -E_{t}\left( v_{t}\right)} \\
&=&I^{0}\left( t,T,X_{t},k\right) -E_{t}\left[ v_{t}\right]  \\
&&+\frac{\rho }{2}E_{t}\left[ H(t,X_{t},k,v_{t})U_{t}(k)\right. 
\\
&&+\frac{\rho ^{2}}{4}E_{t}\left[ \int_{t}^{T}\frac{\partial }{%
\partial x}\left( \frac{\partial ^{2}}{\partial x^{2}}-\frac{\partial }{%
\partial x}\right) H(s,X_{s},k,v_{s})\sigma _{s}\left(
\int_{s}^{T}D_{s}^{W}\sigma _{r}^{2}dr\right) U_{s}(k)ds\right.  \\
&&\left. +\int_{t}^{T}\frac{\partial }{\partial x}H(s,X_{s},k,v_{s})\left( \int_{s}^{T}\left( BS^{-1}\right) ^{\prime }\left( k,\Gamma _{r}\right) \left( D_{s}^{W}\Phi _{r}\right) dr\right) \sigma _{s}ds%
\right]
\end{eqnarray*}%
where $U_{t}(k)=\int_{t}^{T}\left( BS^{-1}\right) ^{\prime }\left( k,\Gamma _{s}\right) \Phi _{s}ds$. 
\end{proposition}

We can now state and prove our main theorem.

\begin{theorem}\label{mainresult}
Consider the model \eqref{themodel} and assume that hypotheses (H1), (H2') and (H3)  hold for some $H\in(0,1) $. Then
\begin{eqnarray*}
\lim_{T\rightarrow t}\frac{I( t,T,X_t,k^+_{t,T})-I( t,T,X_t,k^-_{t,T})  }{(T-t)^{H+\frac12 }} 
&=&\frac{\rho }{2}\lim_{T\rightarrow t}\frac{1}{(T-t)^{H+\frac32}}%
\int_{t}^{T}\left( \int_{s}^{T}D_{s}^{W}\sigma
_{r}^{2}dr\right) ds.
\end{eqnarray*}%
\end{theorem}

From the main theorem the following set of corollaries easily follow. The corollaries contain the results that we will be looking at in the section on numerics.

\begin{corollary}
In the short time-to-maturity limit
\begin{equation}
\lim_{T\rightarrow t}\frac{I( t,T,X_t,k^+_{t,T})-I( t,T,X_t,k^-_{t,T})  }{(T-t)^{H+\frac12 }} 
=\sigma_t^2 \lim_{T\to t}\frac{1}{(T-t)^{H-\frac12}}\frac{\partial I}{\partial k}(k^*_{t,T}).
\end{equation}
Thus, for short time-to-maturity the following approximation holds:
\begin{equation}
I( t,T,X_t,k^+_{t,T})-I( t,T,X_t,k^-_{t,T}) \approx \sigma_t^2 (T-t)\frac{\partial I}{\partial k}(k^*_{t,T}).
\end{equation}
\begin{proof}
Notice that, from the results in Al\`os et al. \cite{ALV} we know that 
$$
\lim_{T\to t}(T-t)^{\frac12-H}\frac{\partial I}{\partial k}(k^*_{t,T})=\frac{\rho }{2\sigma_t^2}\lim_{T\rightarrow t}\frac{1}{(T-t)^{H+\frac32}}%
\int_{t}^{T}\left( \int_{s}^{T}D_{s}^{W}\sigma
_{r}^{2}dr\right) ds.
$$
Jointly with Theorem \ref{mainresult} gives us the result.
\end{proof}
\end{corollary}

\begin{corollary}
In the small time-to-maturity limit
\begin{equation}
\lim_{T\rightarrow t}\frac{I( t,T,X_t,k^+_{t,T})-I( t,T,X_t,k^-_{t,T})  }{(T-t)^{H+\frac12 }} 
= \frac{1}{2\sigma_t}\lim_{T\to t}\frac{1}{(T-t)^{H+\frac12}}E_t\left[\left(\frac{S_T-S_t}{S_t}\right)\frac{1}{(T-t)}\int_t^T \sigma_r^2 dr\right].
\end{equation}
Thus, for short time-to-maturity the following approximation holds:
\begin{equation}
I( t,T,X_t,k^+_{t,T})-I( t,T,X_t,k^-_{t,T})  \approx \frac{1}{2\sigma_t}E_t\left[\left(\frac{S_T-S_t}{S_t}\right)\frac{1}{(T-t)}\int_t^T \sigma_r^2 dr\right].
\end{equation}
\begin{proof}
 Assume the model \eqref{themodel}. Then
 $$
 S_T=S_t+\int_t^T \sigma_rS_r\left( \rho dW_{r}+\sqrt{1-\rho ^{2}}dB_{r}\right).
 $$
 On the other hand, a direct application of the Clark-Ocone-Haussman theorem gives us that (see for example Al\`os and Garc\'ia-Lorite (2024))
 $$
 \int_t^T \sigma_r^2 dr=E_t\left[\int_t^T \sigma_r^2 dr\right]+\int_t^T \int_r^T E_r[D_r^W\sigma_u^2]du dW_r.
 $$
 Then, 
\begin{eqnarray}
        \lefteqn{E_t\left[\left(\frac{S_T-S_t}{S_t}\right)\frac{1}{T-t}\int_t^T \sigma_r^2 dr\right]}\nonumber\\
        &=&\frac{\rho}{S_t} E_t\left[\left(\int_t^T \sigma_rS_rdW_r\right)\frac{1}{T-t}\left(\int_t^T\int_r^T E_r(D_r^W \sigma_u^2)dudW_r\right)\right]\nonumber\\
     &=&\frac{\rho}{S_t} E_t\left[\int_t^T \sigma_rS_r\frac{1}{T-t}\left(\int_r^T E_r[D_r^W \sigma_u^2]du\right)dr\right].
    \end{eqnarray}
Together with Theorem \ref{mainresult} this proves the corollary.
\end{proof}
\end{corollary}

\begin{corollary}\label{covar}
In the small time-to-maturity limit
\begin{equation}\label{covarlim}
\lim_{T\rightarrow t}\frac{I( t,T,X_t,k^+_{t,T})-I( t,T,X_t,k^-_{t,T})  }{(T-t)^{H+\frac12 }} 
=\lim_{T\to t}\frac{1}{(T-t)^{H+\frac12}}E_t\left[\frac{S_T-S_t}{S_t}\sqrt{\frac{1}{T-t}\int_t^T \sigma_r^2 dr}\right].
\end{equation}
That is, for short times to maturity
\begin{equation}\label{covarapprox}
\frac {I( t,T,X_t,k^+_{t,T})-I( t,T,X_t,k^-_{t,T})}{ (T-t)^{H+\frac12}} \approx \frac{1}{(T-t)^{H+\frac12} } E_t\left[\frac{S_T-S_t}{S_t}\sqrt{\frac{1}{T-t}\int_t^T \sigma_r^2 dr}\right].
\end{equation}
\begin{proof}
    A direct application of the Clark-Ocone-Haussman theorem gives us that 
    $$
    \sqrt{\int_t^T \sigma_r^2 dr}=E_t\left[\sqrt{\int_t^T \sigma_r^2 dr}\right] + 
	\int_t^T \int_r^T E_r\left[\frac{D_r^W \sigma_u^2}{2\sqrt{\int_t^T \sigma_r^2 dr}}\right]dudW_r.
   $$
    A similar argument as in the previous corollary allow us to deduce the result.
\end{proof}

\end{corollary}

%%%%%%%%%%%%%%%%%%%%%%%%%%%%%%%%%%%%%%%%%%%%%%%%%%%%%

\section{Numerical examples}

In this section, we examine the accuracy of the approximation derived in Corollary \ref{covar} (which is equivalent to approximation \eqref{rolloos}) and approximation \eqref{Happrox} for $\rho \neq 0$ using the Monte Carlo method under the rough Bergomi model:
\begin{eqnarray}
S_t &=& \exp\left( X_0- \frac{1}{2} \int_{0}^{t} \sigma _{s}^{2} ds+\int_{0}^{t}\sigma
_{s}dB_{s}\right),\\
\sigma_t^2 &=&\sigma_0^2\exp\left(\alpha W_t^H-\frac12 \alpha^2t^{2H}\right)
\end{eqnarray}
where $W_t^H:=\sqrt{2H}\int_0^t (t-s)^{H-\frac12}dW_s$, and for $0\le s < t$ and $\rho \in [-1,1]$,
\begin{eqnarray}
\label{rBcov}
E[W_t^H W_s^H]
&=& 
s^{2H} \int^1_0 \frac{2H}{(1-x)^{\frac12 - H}(t/s-x)^{\frac12 - H}}dx \\
%&=&
%s^{2H}2H\left(\frac{1}{(1-(1/2 -H))(t/s)^{(1/2 -H)}}+ \frac{1-1/2H}{1-(1-1/2H)}\int^1_0 \frac{(1-x)^{1-(1-1/2H)}}{(t/s - x)^{(1-1/2H)+1}}dx\right),\\
\label{rBW}
E[W_t^H B_s]&=&\frac{\rho\sqrt{2H}}{H+\frac12}\left(t^{H+\frac12}-(t-\min(t,s))^{H+\frac12}\right).
\end{eqnarray}

The model parameters are set as $S_0=100$, $\sigma_0=0.2$, $\alpha=0.8$, the correlations are $\rho=-0.2$, $-0.4$, $-0.6$, $-0.8$, the Hurst indices are $H=0.1$, $0.3$, $0.5$, $0.7$, $0.9$, and the maturities are set to $T=0.0025$, $0.005$, $0.01$, $0.025$, $0.05$, $0.1$, $0.25$, $0.5$, $1$, $2$ and $3$.
The number of time steps in the Monte Carlo method is set to $\max\{500T, 100\}$, and the number of simulations is set to 20 million. 

Tables \ref{t1} to \ref{t4} contains the approximation of Corollary \ref{covar}, or equivalently \eqref{rolloos}, for various values of $H$, $\rho$ and $T-t$. Times to maturity less than $0.05$ have not been included to save space. However,  as can also be seen in the tables the accuracy increases as $T-t$ decreases. In the tables ``Cov'' denotes the covariance between $S_T/S_t-1$ and $v_t$

\begin{table}[H]
\centering
\newcolumntype{Y}{>{\centering\arraybackslash}X}
\newcolumntype{Z}{>{\raggedleft\arraybackslash}X}
\begin{tabularx}{\linewidth}{XXYYYYYY} \hline
H index & \multicolumn{1}{c}{Maturity} & $0.05$ & $0.1$ & $0.25$ & $0.5$ & $1$   & $2$ \\
\hline
$0.1$ & $I(k_{t,T}^-)$ & 0.1977 & 0.1974 & 0.1969 & 0.1963 & 0.1958 & 0.1952 \\
      & $I(k_{t,T}^+)$ & 0.1975 & 0.1970 & 0.1963 & 0.1954 & 0.1945 & 0.1931 \\
      & Cov & -0.0002 & -0.0004 & -0.0006 & -0.0009 & -0.0014 & -0.0022 \\
      & $\frac{I(k_{t,T}^+) - I(k_{t,T}^-)}{(T-t)^{H+1/2}}$ & -0.0014 & -0.0014 & -0.0014 & -0.0014 & -0.0014 & -0.0013 \\
      & $\frac{\mathrm{Cov}}{(T-t)^{H+1/2}}$ & -0.0015 & -0.0014 & -0.0015	& -0.0014	& -0.0014	& -0.0014 \\
\hline
$0.3$ & $I(k_{t,T}^-)$ & 0.1990 & 0.1985 & 0.1975 & 0.1961 & 0.1941 & 0.1911 \\
      & $I(k_{t,T}^+)$ & 0.1989 & 0.1983 & 0.1969 & 0.1952 & 0.1926 & 0.1885 \\
      & Cov & -0.0002 & -0.0003 & -0.0006 & -0.0010 & -0.0017 & -0.0029 \\
      & $\frac{I(k_{t,T}^+) - I(k_{t,T}^-)}{(T-t)^{H+1/2}}$ & -0.0017	& -0.0017	& -0.0016	& -0.0016	& -0.0015	& -0.0015 \\
      & $\frac{\mathrm{Cov}}{(T-t)^{H+1/2}}$ & -0.0017	& -0.0017	& -0.0017	& -0.0017	& -0.0017	& -0.0016 \\
\hline
$0.5$ & $I(k_{t,T}^-)$ & 0.1997 & 0.1995 & 0.1987 & 0.1973 & 0.1946 & 0.1893 \\
      & $I(k_{t,T}^+)$ & 0.1997 & 0.1993 & 0.1983 & 0.1965 & 0.1932 & 0.1867 \\
      & Cov & -0.0001 & -0.0002 & -0.0004 & -0.0008 & -0.0016 & -0.0030 \\
      & $\frac{I(k_{t,T}^+) - I(k_{t,T}^-)}{(T-t)^{H+1/2}}$ & -0.0016	& -0.0016	& -0.0015	& -0.0015	& -0.0014	& -0.0013 \\
      & $\frac{\mathrm{Cov}}{(T-t)^{H+1/2}}$ & -0.0016	& -0.0016	& -0.0016	& -0.0016	& -0.0016	& -0.0015 \\
\hline
$0.7$ & $I(k_{t,T}^-)$ & 0.1999 & 0.1998 & 0.1993 & 0.1982 & 0.1954 & 0.1881 \\
      & $I(k_{t,T}^+)$ & 0.1999 & 0.1997 & 0.1991 & 0.1976 & 0.1941 & 0.1855 \\
      & Cov & 0.0000 & -0.0001 & -0.0003 & -0.0006 & -0.0014 & -0.0031 \\
      & $\frac{I(k_{t,T}^+) - I(k_{t,T}^-)}{(T-t)^{H+1/2}}$ & -0.0014	& -0.0014	& -0.0014	& -0.0014	& -0.0013	& -0.0011 \\
      & $\frac{\mathrm{Cov}}{(T-t)^{H+1/2}}$ & -0.0014	& -0.0014	& -0.0014	& -0.0014	& -0.0014	& -0.0013 \\
\hline
$0.9$ & $I(k_{t,T}^-)$ & 0.2000 & 0.1999 & 0.1997 & 0.1989 & 0.1961 & 0.1869 \\
      & $I(k_{t,T}^+)$ & 0.2000 & 0.1999 & 0.1995 & 0.1984 & 0.1949 & 0.1845 \\
      & Cov & 0.0000 & -0.0001 & -0.0002 & -0.0005 & -0.0012 & -0.0031 \\
      & $\frac{I(k_{t,T}^+) - I(k_{t,T}^-)}{(T-t)^{H+1/2}}$ & -0.0012	& -0.0012	& -0.0013	& -0.0012	& -0.0012	& -0.0009 \\
      & $\frac{\mathrm{Cov}}{(T-t)^{H+1/2}}$ & -0.0013	& -0.0013	& -0.0013	& -0.0013	& -0.0012	& -0.0012 \\
\hline
\end{tabularx}%
\caption{Approximate error of covariance in $\rho=-0.2$}
\label{t1}
\end{table}

\begin{table}[H]
\centering
\newcolumntype{Y}{>{\centering\arraybackslash}X}
\newcolumntype{Z}{>{\raggedleft\arraybackslash}X}
\begin{tabularx}{\linewidth}{XXYYYYYY} \hline
H index & \multicolumn{1}{c}{Maturity} & $0.05$ & $0.1$ & $0.25$ & $0.5$ & $1$   & $2$ \\
\hline
$0.1$ & $I(k_{t,T}^-)$ & 0.1975 & 0.1971 & 0.1966 & 0.1960 & 0.1954 & 0.1947 \\
      & $I(k_{t,T}^+)$ & 0.1970 & 0.1964 & 0.1954 & 0.1942 & 0.1927 & 0.1907 \\
      & Cov & -0.0005 & -0.0007 & -0.0013 & -0.0019 & -0.0028 & -0.0042 \\
      & $\frac{I(k_{t,T}^+) - I(k_{t,T}^-)}{(T-t)^{H+1/2}}$ & -0.0028 & -0.0028	& -0.0027	& -0.0027	& -0.0027	& -0.0026 \\
      & $\frac{\mathrm{Cov}}{(T-t)^{H+1/2}}$ & -0.0029	& -0.0029	& -0.0029	& -0.0029	& -0.0028	& -0.0028 \\
\hline
$0.3$ & $I(k_{t,T}^-)$ & 0.1989 & 0.1984 & 0.1972 & 0.1957 & 0.1936 & 0.1902 \\
      & $I(k_{t,T}^+)$ & 0.1986 & 0.1979 & 0.1962 & 0.1939 & 0.1905 & 0.1852 \\
      & Cov & -0.0003 & -0.0005 & -0.0011 & -0.0019 & -0.0033 & -0.0056 \\
      & $\frac{I(k_{t,T}^+) - I(k_{t,T}^-)}{(T-t)^{H+1/2}}$ & -0.0033	& -0.0033	& -0.0032	& -0.0032	& -0.0030	& -0.0029 \\
      & $\frac{\mathrm{Cov}}{(T-t)^{H+1/2}}$& -0.0034	& -0.0034	& -0.0034	& -0.0033	& -0.0033	& -0.0032 \\
\hline
$0.5$ & $I(k_{t,T}^-)$ & 0.1997 & 0.1994 & 0.1985 & 0.1970 & 0.1941 & 0.1883 \\
      & $I(k_{t,T}^+)$ & 0.1996 & 0.1991 & 0.1978 & 0.1955 & 0.1913 & 0.1833 \\
      & Cov & -0.0002 & -0.0003 & -0.0008 & -0.0016 & -0.0031 & -0.0059 \\
      & $\frac{I(k_{t,T}^+) - I(k_{t,T}^-)}{(T-t)^{H+1/2}}$ & -0.0031	& -0.0031	& -0.0031	& -0.0030	& -0.0028	& -0.0025 \\
      & $\frac{\mathrm{Cov}}{(T-t)^{H+1/2}}$ & -0.0032	& -0.0032	& -0.0032	& -0.0031	& -0.0031	& -0.0030 \\
\hline
$0.7$ & $I(k_{t,T}^-)$ & 0.1999 & 0.1998 & 0.1993 & 0.1981 & 0.1950 & 0.1870 \\
      & $I(k_{t,T}^+)$ & 0.1998 & 0.1996 & 0.1988 & 0.1969 & 0.1924 & 0.1821 \\
      & Cov & -0.0001 & -0.0002 & -0.0005 & -0.0012 & -0.0028 & -0.0060 \\
      & $\frac{I(k_{t,T}^+) - I(k_{t,T}^-)}{(T-t)^{H+1/2}}$ & -0.0028	& -0.0028	& -0.0028	& -0.0027	& -0.0026	& -0.0022 \\
      & $\frac{\mathrm{Cov}}{(T-t)^{H+1/2}}$ & -0.0028	& -0.0028	& -0.0028	& -0.0028	& -0.0028	& -0.0026 \\
\hline
$0.9$ & $I(k_{t,T}^-)$ & 0.2000 & 0.1999 & 0.1997 & 0.1988 & 0.1958 & 0.1858 \\
      & $I(k_{t,T}^+)$ & 0.1999 & 0.1998 & 0.1993 & 0.1978 & 0.1935 & 0.1810 \\
      & Cov & 0.0000 & -0.0001 & -0.0004 & -0.0009 & -0.0025 & -0.0060 \\
      & $\frac{I(k_{t,T}^+) - I(k_{t,T}^-)}{(T-t)^{H+1/2}}$ & -0.0025	& -0.0025	& -0.0025	& -0.0025	& -0.0023	& -0.0018 \\
      & $\frac{\mathrm{Cov}}{(T-t)^{H+1/2}}$ & -0.0025	& -0.0025	& -0.0025	& -0.0025	& -0.0025	& -0.0023 \\
\hline
\end{tabularx}%
\caption{Approximate error of covariance in $\rho=-0.4$}
\label{t2}
\end{table}

\begin{table}[H]
\centering
\newcolumntype{Y}{>{\centering\arraybackslash}X}
\newcolumntype{Z}{>{\raggedleft\arraybackslash}X}
\begin{tabularx}{\linewidth}{XXYYYYYY} \hline
H index & \multicolumn{1}{c}{Maturity} & $0.05$ & $0.1$ & $0.25$ & $0.5$ & $1$   & $2$ \\
\hline
$0.1$ & $I(k_{t,T}^-)$ & 0.1971 & 0.1967 & 0.1961 & 0.1954 & 0.1948 & 0.1939 \\
      & $I(k_{t,T}^+)$ & 0.1964 & 0.1957 & 0.1943 & 0.1928 & 0.1908 & 0.1881 \\
      & Cov & -0.0007 & -0.0011 & -0.0019 & -0.0028 & -0.0042 & -0.0062 \\
      & $\frac{I(k_{t,T}^+) - I(k_{t,T}^-)}{(T-t)^{H+1/2}}$ & -0.0042	& -0.0042	& -0.0041	& -0.0040	& -0.0040	& -0.0039 \\
      & $\frac{\mathrm{Cov}}{(T-t)^{H+1/2}}$ & -0.0043	& -0.0043	& -0.0043	& -0.0043	& -0.0042	& -0.0041 \\
\hline
$0.3$ & $I(k_{t,T}^-)$ & 0.1988 & 0.1982 & 0.1968 & 0.1952 & 0.1927 & 0.1889 \\
      & $I(k_{t,T}^+)$ & 0.1983 & 0.1974 & 0.1952 & 0.1925 & 0.1882 & 0.1816 \\
      & Cov & -0.0005 & -0.0008 & -0.0017 & -0.0029 & -0.0049 & -0.0082 \\
      & $\frac{I(k_{t,T}^+) - I(k_{t,T}^-)}{(T-t)^{H+1/2}}$ & -0.0050	& -0.0050	& -0.0048	& -0.0047	& -0.0045	& -0.0042 \\
      & $\frac{\mathrm{Cov}}{(T-t)^{H+1/2}}$ & -0.0051	& -0.0051	& -0.0050	& -0.0050	& -0.0049	& -0.0047 \\
\hline
$0.5$ & $I(k_{t,T}^-)$ & 0.1997 & 0.1993 & 0.1983 & 0.1966 & 0.1933 & 0.1868 \\
      & $I(k_{t,T}^+)$ & 0.1994 & 0.1989 & 0.1972 & 0.1944 & 0.1892 & 0.1794 \\
      & Cov & -0.0002 & -0.0005 & -0.0012 & -0.0023 & -0.0046 & -0.0086 \\
      & $\frac{I(k_{t,T}^+) - I(k_{t,T}^-)}{(T-t)^{H+1/2}}$ & -0.0047	& -0.0047	& -0.0046	& -0.0045	& -0.0042	& -0.0037 \\
      & $\frac{\mathrm{Cov}}{(T-t)^{H+1/2}}$ & -0.0047	& -0.0047	& -0.0047	& -0.0047	& -0.0046	& -0.0043 \\
\hline
$0.7$ & $I(k_{t,T}^-)$ & 0.1999 & 0.1998 & 0.1992 & 0.1978 & 0.1944 & 0.1854 \\
      & $I(k_{t,T}^+)$ & 0.1998 & 0.1995 & 0.1984 & 0.1961 & 0.1906 & 0.1782 \\
      & Cov & -0.0001 & -0.0003 & -0.0008 & -0.0018 & -0.0041 & -0.0087 \\
      & $\frac{I(k_{t,T}^+) - I(k_{t,T}^-)}{(T-t)^{H+1/2}}$ & -0.0042	& -0.0042	& -0.0042	& -0.0041	& -0.0038	& -0.0032 \\
      & $\frac{\mathrm{Cov}}{(T-t)^{H+1/2}}$ & -0.0043	& -0.0042	& -0.0042	& -0.0042	& -0.0041	& -0.0038 \\
\hline
$0.9$ & $I(k_{t,T}^-)$ & 0.2000 & 0.1999 & 0.1996 & 0.1986 & 0.1953 & 0.1842 \\
      & $I(k_{t,T}^+)$ & 0.1999 & 0.1998 & 0.1991 & 0.1972 & 0.1919 & 0.1771 \\
      & Cov & -0.0001 & -0.0002 & -0.0005 & -0.0014 & -0.0036 & -0.0087 \\
      & $\frac{I(k_{t,T}^+) - I(k_{t,T}^-)}{(T-t)^{H+1/2}}$ & -0.0038	& -0.0038	& -0.0037	& -0.0037	& -0.0034	& -0.0027 \\
      & $\frac{\mathrm{Cov}}{(T-t)^{H+1/2}}$ & -0.0038	& -0.0038	& -0.0038	& -0.0038	& -0.0036	& -0.0033 \\
\hline
\end{tabularx}%
\caption{Approximate error of covariance in $\rho=-0.6$}
\label{t3}
\end{table}

\begin{table}[H]
\centering
\newcolumntype{Y}{>{\centering\arraybackslash}X}
\newcolumntype{Z}{>{\raggedleft\arraybackslash}X}
\begin{tabularx}{\linewidth}{XXYYYYYY} \hline
H index & \multicolumn{1}{c}{Maturity} & $0.05$ & $0.1$ & $0.25$ & $0.5$ & $1$   & $2$ \\
\hline
$0.1$ & $I(k_{t,T}^-)$ & 0.1966 & 0.1961 & 0.1954 & 0.1946 & 0.1939 & 0.1929 \\
      & $I(k_{t,T}^+)$ & 0.1957 & 0.1948 & 0.1931 & 0.1912 & 0.1887 & 0.1853 \\
      & Cov & -0.0010 & -0.0014 & -0.0025 & -0.0037 & -0.0056 & -0.0082 \\
      & $\frac{I(k_{t,T}^+) - I(k_{t,T}^-)}{(T-t)^{H+1/2}}$ & -0.0056	& -0.0055	& -0.0054	& -0.0053	& -0.0052	& -0.0050 \\
      & $\frac{\mathrm{Cov}}{(T-t)^{H+1/2}}$ & -0.0058	& -0.0057	& -0.0057	& -0.0056	& -0.0056	& -0.0054 \\
\hline
$0.3$ & $I(k_{t,T}^-)$ & 0.1986 & 0.1979 & 0.1963 & 0.1943 & 0.1915 & 0.1871 \\
      & $I(k_{t,T}^+)$ & 0.1980 & 0.1968 & 0.1942 & 0.1908 & 0.1856 & 0.1778 \\
      & Cov & -0.0006 & -0.0011 & -0.0022 & -0.0038 & -0.0064 & -0.0107 \\
      & $\frac{I(k_{t,T}^+) - I(k_{t,T}^-)}{(T-t)^{H+1/2}}$ & -0.0067	& -0.0066	& -0.0064	& -0.0062	& -0.0059	& -0.0054 \\
      & $\frac{\mathrm{Cov}}{(T-t)^{H+1/2}}$ & -0.0068	& -0.0067	& -0.0067	& -0.0066	& -0.0064	& -0.0061 \\
\hline
$0.5$ & $I(k_{t,T}^-)$ & 0.1996 & 0.1992 & 0.1981 & 0.1961 & 0.1923 & 0.1848 \\
      & $I(k_{t,T}^+)$ & 0.1993 & 0.1986 & 0.1965 & 0.1932 & 0.1868 & 0.1753 \\
      & Cov & -0.0003 & -0.0006 & -0.0016 & -0.0031 & -0.0060 & -0.0112 \\
      & $\frac{I(k_{t,T}^+) - I(k_{t,T}^-)}{(T-t)^{H+1/2}}$ & -0.0063	& -0.0062	& -0.0061	& -0.0059	& -0.0055	& -0.0048 \\
      & $\frac{\mathrm{Cov}}{(T-t)^{H+1/2}}$ & -0.0063	& -0.0063	& -0.0063	& -0.0062	& -0.0060	& -0.0056 \\
\hline
$0.7$ & $I(k_{t,T}^-)$ & 0.1999 & 0.1997 & 0.1991 & 0.1975 & 0.1935 & 0.1833 \\
      & $I(k_{t,T}^+)$ & 0.1997 & 0.1994 & 0.1980 & 0.1952 & 0.1885 & 0.1739 \\
      & Cov & -0.0002 & -0.0004 & -0.0011 & -0.0024 & -0.0054 & -0.0114 \\
      & $\frac{I(k_{t,T}^+) - I(k_{t,T}^-)}{(T-t)^{H+1/2}}$ & -0.0057	& -0.0056	& -0.0056	& -0.0054	& -0.0050	& -0.0041 \\
      & $\frac{\mathrm{Cov}}{(T-t)^{H+1/2}}$ & -0.0057	& -0.0057	& -0.0056	& -0.0056	& -0.0054	& -0.0049 \\
\hline
$0.9$ & $I(k_{t,T}^-)$ & 0.2000 & 0.1999 & 0.1996 & 0.1984 & 0.1945 & 0.1819 \\
      & $I(k_{t,T}^+)$ & 0.1999 & 0.1997 & 0.1988 & 0.1966 & 0.1901 & 0.1728 \\
      & Cov & -0.0001 & -0.0002 & -0.0007 & -0.0019 & -0.0048 & -0.0113 \\
      & $\frac{I(k_{t,T}^+) - I(k_{t,T}^-)}{(T-t)^{H+1/2}}$ & -0.0050	& -0.0050	& -0.0050	& -0.0049	& -0.0045	& -0.0034 \\
      & $\frac{\mathrm{Cov}}{(T-t)^{H+1/2}}$ & -0.0050	& -0.0050	& -0.0050	& -0.0050	& -0.0048	& -0.0043 \\
\hline
\end{tabularx}%
\caption{Approximate error of covariance in $\rho=-0.8$}
\label{t4}
\end{table}

In addition to the tables, in Figure \ref{fig1} we plot the ratios $(I(k_{t,T}^+) - I(k_{t,T}^-))/\mathrm{Cov}$ for various values of $H$ and $T-t$ while keeping correlation fixed at $\rho = -0.8$. In Figure \eqref{fig2} the Hurst parameter is fixed $H=0.3$ and $T-t$ and $\rho$ are varied. We observe that the Hurst parameter and time to maturity has a larger impact on accuracy than correlation. The sensitivity of the accuracy to the Hurst parameter is consistent with the fact that the error of the approximation in Corollary \ref{covar} is $O((T-t)^{2H+1})$ (see proof of Theorem \ref{mainresult} in Appendix). For small values of $H$ the error is less for $T-t >1$ but converges slower for $T-t <1$.

\begin{figure}[H]
\centering
\begin{tikzpicture}[scale=0.8]
\begin{axis}[
	legend style={font=\small},
	legend pos= south west,
	xlabel= $T-t$, 
	ylabel= $(I(k_{t,T}^+) - I(k_{t,T}^-))/\mathrm{Cov}$,
	ymin = .75,
	title= {$\rho = -0.8$},
	grid=both,
	minor grid style={gray!25},
	major grid style={gray!25},
	width=0.75\linewidth,
	y tick label style={
        /pgf/number format/.cd,
        fixed,
        fixed zerofill,
        precision=2,
        /tikz/.cd
    },]
	
%\addplot[line width=1pt,solid,color=red] 
%	table[x=x, y=y1, col sep=tab]{fig1.txt};
%\addlegendentry{$H=0.1$};	
\addlegendimage{only marks, mark=x}
\addlegendimage{only marks, mark=o}
\addlegendimage{only marks, mark=+}

\addplot[mark = x, color=black] 
	table[x=x, y=y1, col sep=tab]{fig1.txt};
\addlegendentry{$H=0.1$};	

\addplot[mark = o, color=black] 
	table[x=x, y=y2, col sep=tab]{fig1.txt};
\addlegendentry{$H=0.3$};	

\addplot[mark = +, color=black] 
	table[x=x, y=y3, col sep=tab]{fig1.txt};
\addlegendentry{$H=0.5$};	

\end{axis}
\end{tikzpicture}
\caption{Impact of $H$ and $T-t$ on covariance estimate accuracy}
\label{fig1}
\end{figure}
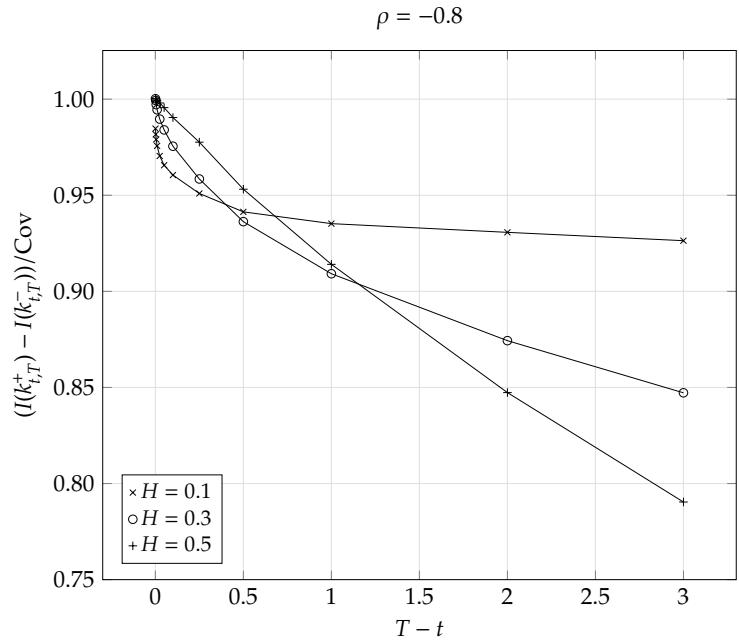

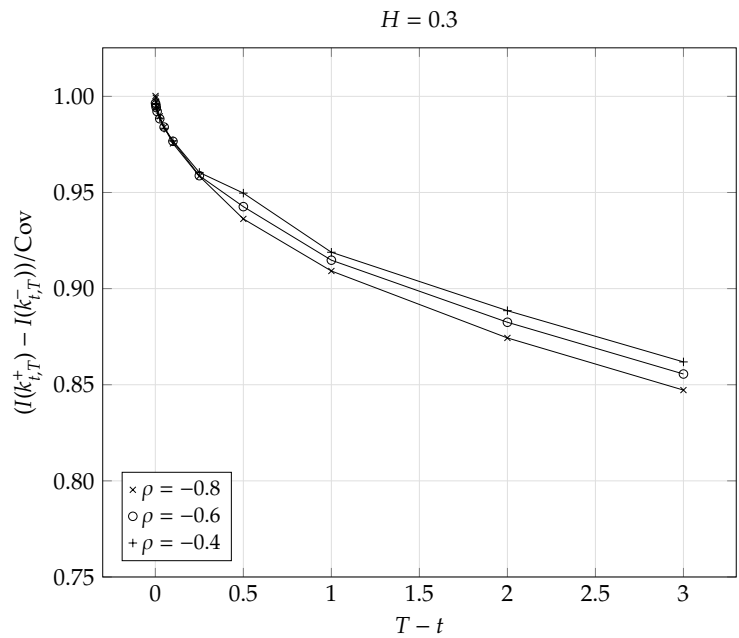
\begin{figure}[H]
\centering
\begin{tikzpicture}[scale=0.8]
\begin{axis}[
	legend style={font=\small},
	legend pos= south west,
	xlabel= $T-t$, 
	ylabel= $(I(k_{t,T}^+) - I(k_{t,T}^-))/\mathrm{Cov}$,
	ymin = .75,
	title= {$H = 0.3$},
	grid=both,
	minor grid style={gray!25},
	major grid style={gray!25},
	width=0.75\linewidth,
	y tick label style={
        /pgf/number format/.cd,
        fixed,
        fixed zerofill,
        precision=2,
        /tikz/.cd
    },]
	
\addlegendimage{only marks, mark=x}
\addlegendimage{only marks, mark=o}
\addlegendimage{only marks, mark=+}

\addplot[mark = x, color=black] 
	table[x=x, y=y1, col sep=tab]{fig2.txt};
\addlegendentry{$\rho=-0.8$};	

\addplot[mark = o, color=black] 
	table[x=x, y=y2, col sep=tab]{fig2.txt};
\addlegendentry{$\rho=-0.6$};	

\addplot[mark = +, color=black] 
	table[x=x, y=y3, col sep=tab]{fig2.txt};
\addlegendentry{$\rho=-0.4$};	

\end{axis}
\end{tikzpicture}
\caption{Impact of $\rho$ and $T-t$ on covariance estimate accuracy}
\label{fig2}
\end{figure}

Lastly we examine the accuracy of approximation \eqref{Happrox} by plotting the estimated value of $H$ as given by expression \eqref{Happrox} divided by its exact value. Equation \eqref{Happrox} requires two value for time to maturity. In the plots we have fixed $T_1 = 0.0025$ and let $T_2 \in \{0.005, 0.01, 0.025, 0.05, 0.1, 0.25, 0.5, 1, 2, 3 \}$. In figure \ref{fig3} the accuracy is shown for $\rho = -0.8$ and various values for $H$ and $T_2 -t$. In figure \ref{fig4} we fix $H=0.3$ and plot the ratio for different $\rho$ and $T_2-t$. We observe that for especially for short maturities the simple approximation \eqref{Happrox} is accurate.

\begin{figure}[H]
\centering
\begin{tikzpicture}[scale=.8]
\begin{axis}[
	legend style={font=\small},
	legend pos= south west,
	xlabel= $T-t$, 
	ylabel= Eq. \eqref{Happrox}/Exact value,
	ymin = .9,
	title= {$\rho = -0.8$},
	grid=both,
	minor grid style={gray!25},
	major grid style={gray!25},
	width=0.75\linewidth,
	y tick label style={
        /pgf/number format/.cd,
        fixed,
        fixed zerofill,
        precision=2,
        /tikz/.cd
    },]
	
\addlegendimage{only marks, mark=x}
\addlegendimage{only marks, mark=o}
\addlegendimage{only marks, mark=+}

\addplot[mark = x, color=black] 
	table[x=x, y=y1, col sep=tab]{fig3.txt};
\addlegendentry{$H=0.1$};	

\addplot[mark = o, color=black] 
	table[x=x, y=y2, col sep=tab]{fig3.txt};
\addlegendentry{$H=0.3$};	

\addplot[mark = +, color=black] 
	table[x=x, y=y3, col sep=tab]{fig3.txt};
\addlegendentry{$H=0.5$};	

\end{axis}
\end{tikzpicture}
\caption{Impact of $H$ and $T-t$ on $H$ estimate accuracy}
\label{fig3}
\end{figure}
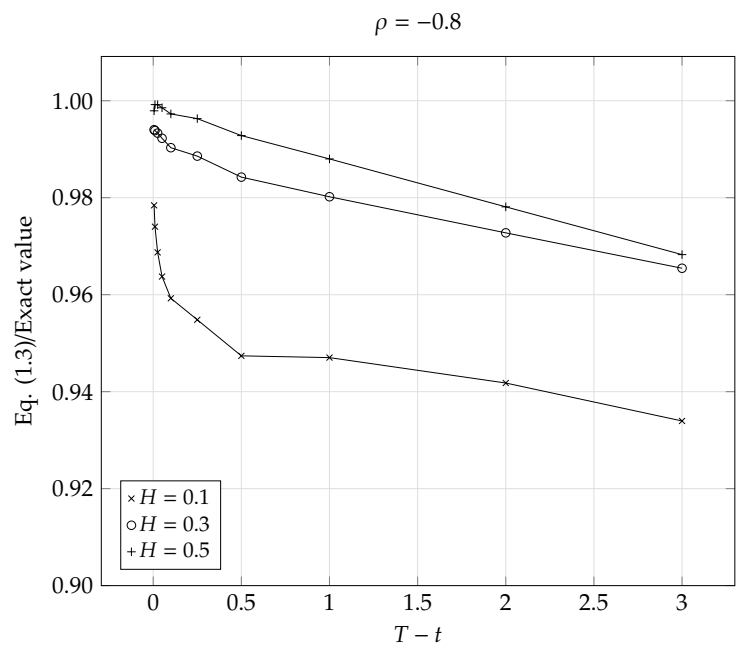

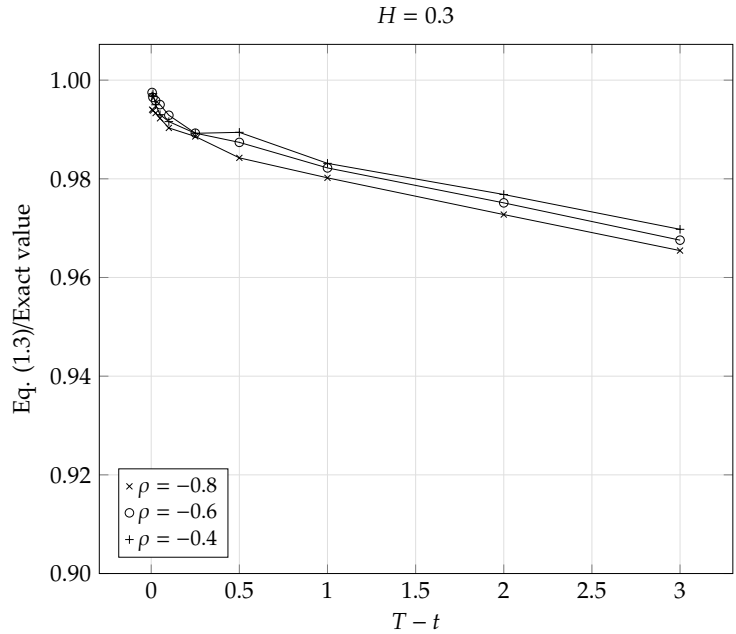
\begin{figure}[H]
\centering
\begin{tikzpicture}[scale=0.8]
\begin{axis}[
	legend style={font=\small},
	legend pos= south west,
	xlabel= $T-t$, 
	ylabel= Eq. \eqref{Happrox}/Exact value,
	ymin = .9,
	title= {$H=0.3$},
	grid=both,
	minor grid style={gray!25},
	major grid style={gray!25},
	width=0.75\linewidth,
	y tick label style={
        /pgf/number format/.cd,
        fixed,
        fixed zerofill,
        precision=2,
        /tikz/.cd
    },]
	
\addlegendimage{only marks, mark=x}
\addlegendimage{only marks, mark=o}
\addlegendimage{only marks, mark=+}

\addplot[mark = x, color=black] 
	table[x=x, y=y1, col sep=tab]{fig4.txt};
\addlegendentry{$\rho=-0.8$};	

\addplot[mark = o, color=black] 
	table[x=x, y=y2, col sep=tab]{fig4.txt};
\addlegendentry{$\rho=-0.6$};	

\addplot[mark = +, color=black] 
	table[x=x, y=y3, col sep=tab]{fig4.txt};
\addlegendentry{$\rho=-0.4$};	

\end{axis}
\end{tikzpicture}
\caption{Impact of $\rho$ and $T-t$ on $H$ estimate accuracy}
\label{fig4}
\end{figure}

\section{Conclusion}
In this paper we have derived rigorous limit theorems for the approximation \eqref{rolloos}. In addition we have noted that it can be used to estimate the value of the Hurst parameter as given by expression \eqref{Happrox}. Numerical runs have confirmed, under the rough Bergomi model, that the approximations are accurate for short times to maturity for different values of the Hurst parameter and correlation. The results are not only of theoretical interest, but also of practical interest. They can, for instance, be used to calibrate or estimate the Hurst parameter from a limited number of short dated options.

\appendix

\section{Greeks}\label{appendix2}

This section shows the proofs of Propositions and Theorems in Section \ref{sec3}.
Firstly, we give some Greeks of Black-Scholes formula.

A direct calculation gives us that  $k\in\mathbb{R}$ and all $u>0$:
$$
(BS^{-1})'(k,u)=\frac{1}{\frac{\partial BS}{\partial\sigma}(k,BS^{-1}(k,u))}.
$$
Then it follows that 
\begin{eqnarray}
(BS^{-1})''(k,u)&=&-\frac{1}{(\frac{\partial BS}{\partial\sigma}(k,BS^{-1}(k,u)))^2}\frac{\partial^2 BS}{\partial\sigma^2}(k,BS^{-1}(k,u))\frac{1}{\frac{\partial BS}{\partial\sigma}(k,BS^{-1}(k,u))}\nonumber\\
&=&-\frac{1}{(\frac{\partial BS}{\partial\sigma}(k,BS^{-1}(k,u)))^3}\frac{\partial^2 BS}{\partial\sigma^2}(k,BS^{-1}(k,u)).
\end{eqnarray}
Now, the classical relationship between the {\it Vomma} and the {\it Vega}
$$
\frac{\partial^2 BS}{\partial\sigma^2}(k,\sigma)=\frac{\partial BS}{\partial\sigma}(k,\sigma)\frac{d_1(k,\sigma)d_2(k,\sigma)}{\sigma}
$$
allows us to write
$$
(BS^{-1})''(k,u)=\frac{1}{(\frac{\partial BS}{\partial\sigma}(k,BS^{-1}(k,u)))^2}\frac{(BS^{-1}(k,u))^4(T-t)^2-4(X_t-k)^2}{4(BS^{-1}(k,u))^3(T-t)}.
$$
Finally, as
$$\frac{\partial BS}{\partial\sigma}(k,BS^{-1}(k,u))=
\exp(X_{t})N^{\prime }(d_1\left( k,BS^{-1}(k,u)\right))\sqrt{T-t},$$ the above equality reduces to
\begin{equation}
\label{greek}
(BS^{-1})''(k,u)=\frac{(BS^{-1}(k,u))^4(T-t)^2-4(X_t-k)^2}{4(\exp(X_{t})N^{\prime }(d_1\left( k,BS^{-1}(k,u)\right))(T-t))^2(BS^{-1}(k,u))^3}.
\end{equation}

\section{Proofs}\label{appendix3}

\begin{proof}[Proof of Proposition \ref{newdecomposition}] 
Proposition \ref{Theoremcorrelatedcase} gives us that
\begin{eqnarray*}
\lefteqn{I\left( t,T,X_{t},k\right) -E_{t}\left[ v_{t}\right]}  \\
&=&I^{0}\left( t,T,X_{t},k\right) -E_{t}\left[ v_{t}\right]  \\
&&+\frac{\rho }{2}E_{t}\int_{t}^{T}\left( BS^{-1}\right) ^{\prime }\left(
k,\Gamma _{s}\right) H(s,X_{s},k,v_{s})\Phi _{s}ds \\
&=&T_{1}+T_{2}.
\end{eqnarray*}%
Now we
apply the anticipating It\^{o}'s formula (see for example Nualart (2005)) to
the process%
\[
H(t,X_{t},k,v_{t})U_{t}(k),
\]%
and we get%
\begin{eqnarray*}
0 &=&E_{t}\Bigg[ H(t,X_{t},k,v_{t})U_{t}(k) \\
&&-\int_{t}^{T}H(s,X_{s},k,v_{s})\left( BS^{-1}\right) ^{\prime
}\left( k,\Gamma _{s}\right) \Phi _{s}ds \\
&&+\int_{t}^{T}\frac{\partial H}{\partial s}(s,X_{s},k,v_{s})U_{s}(k)ds
\\
&&+\frac{\rho}{2}\int_{t}^{T}D_{s}^{W}\left( \frac{\partial H}{\partial x}%
(s,X_{s},k,v_{s})U_{s}(k)\right) \sigma _{s}ds \\
&&+\frac{1}{2}\int_{t}^{T}\frac{\partial ^{2}H}{\partial x^{2}}%
(s,X_{s},k,v_{s})\sigma _{s}^{2}U_{s}(k)ds \\
&& +\frac{1}{2}\int_{t}^{T}\left( \frac{\partial ^{2}}{\partial x^{2}}-%
\frac{\partial }{\partial x}\right) H(s,X_{s},k,v_{s})\left(
v_{s}^{2}-\sigma _{s}^{2}\right) U_{s}(k)ds\Bigg],
\end{eqnarray*}%
which implies that%
\begin{eqnarray*}
T_{2} &=&\frac{\rho }{2}E_{t}\Bigg[ H(t,X_{t},k,v_{t})U_{t}
\\
&&+\frac{\rho}{2}\int_{t}^{T}D_{s}^{W}\left( \frac{\partial H}{%
\partial x}(s,X_{s},k,v_{s})U_{s}(k)\right) \sigma _{s}ds\Bigg]  \\
&=&T_{2}^{1}+T_{2}^{2}.
\end{eqnarray*}%
Now, notice that
\begin{eqnarray*}
T_{2}^{2} &=&\frac{\rho ^{2}}{4}E_{t}\left[ \int_{t}^{T}\frac{\partial }{%
\partial x}\left( \frac{\partial ^{2}}{\partial x^{2}}-\frac{\partial }{%
\partial x}\right) H(s,X_{s},k,v_{s})\sigma _{s}\left(
\int_{s}^{T}D_{s}^{W}\sigma _{r}^{2}dr\right) U_{s}(k)ds\right.  \\
&&\left. +\int_{t}^{T}\frac{\partial }{\partial x}H(s,X_{s},k,v_{s})\left( \int_{s}^{T}\left( BS^{-1}\right) ^{\prime }\left( k,\Gamma _{r}\right) \left( D_{s}^{W}\Phi _{r}\right) dr\right) \sigma _{s}ds%
\right].
\end{eqnarray*}%
Now the proof is complete.
\end{proof}

\begin{proof}[Proof of Theorem \ref{mainresult}] 
Proposition \ref{newdecomposition} allows us to write
\begin{eqnarray}
\lefteqn{I\left( t,T,X_{t},k^+_{t,T}\right) -I\left( t,T,X_{t},k^-_{t,T}\right)}\nonumber  \\
&=&I^{0}
( t,T,X_{t},k^+_{t,T})
-I^{0}( t,T,X_{t},k^-_{t,T}) \nonumber\\
&&+\frac{\rho }{2}E_{t}\left[ H(t,X_{t},k^+_{t,T},v_{t})U_{t}(k_t^+)-H(t,X_{t},k^-_{t,T},v_{t})U_{t}(k_t^-)\right]\nonumber\\
&&+\frac{\rho ^{2}}{4}E_{t}\left[ \int_{t}^{T}\frac{\partial }{%
\partial x}\left( \frac{\partial ^{2}}{\partial x^{2}}-\frac{\partial }{%
\partial x}\right) H(s,X_{s},k^+_{t,T},v_{s})\sigma _{s}\left(
\int_{s}^{T}D_{s}^{W}\sigma _{r}^{2}dr\right) U_{s}(k^+_{t,T})ds\right.  \nonumber\\
&&-\int_{t}^{T}\frac{\partial }{%
\partial x}\left( \frac{\partial ^{2}}{\partial x^{2}}-\frac{\partial }{%
\partial x}\right) H(s,X_{s},k^-_{t,T},v_{s})\sigma _{s}\left(\int_{s}^{T}D_{s}^{W}\sigma _{r}^{2}dr\right) U_{s}(k^-_{t,T})ds\nonumber\\
&& +2\int_{t}^{T}\frac{\partial }{\partial x}H(s,X_{s},k^+_{t,T},v_{s})\left( \int_{s}^{T}\left( BS^{-1}\right) ^{\prime }\left(k^+_{t,T},\Gamma _{r}\right) \left( D_{s}^{W}\Phi _{r}\right) dr\right) \sigma _{s}ds\nonumber\\
&&\left. -2\int_{t}^{T}
\frac{\partial }{\partial x}H(s,X_{s},k^-_{t,T},v_{s})\left( \int_{s}^{T}\left( BS^{-1}\right) ^{\prime }\left( k^-_{t,T},\Gamma _{r}\right) \left( D_{s}^{W}\Phi _{r}\right) dr\right) \sigma _{s}ds\right].
\label{superdecomposition}
\end{eqnarray}%

Now the proof is decomposed into several steps.

{\it Step 1} By making use of Proposition 3.1 from Renault and Touzi \cite{RT} it can readily be seen that when $\rho=0$:
$$I^0( t,T,X_t,k^+_{t,T}) = I^0( t,T,X_t,k^-_{t,T}).$$

{\it Step 2} As 
\[
H(t,X_t,k,v_t)=\frac{e^{X_t}N^{\prime }(d_{+}\left(
k,v_t\right) )}{v_t\sqrt{T-t}}\left( 1-\frac{d_{+}\left( k,v_t\right) 
}{v_t\sqrt{T-t}}\right) .
\]%
it follows that 
\[
H(t,X_t,k^+_{t,T},v_t)=\frac{e^{X_t}N^{\prime }(d_{+}\left(
k^+_{t,T},v_t\right) )}{2v_t\sqrt{T-t}}
\left( \frac{1}{v_t^2(T-t)} \left(v_t^2+I(k^+_{t,T})^2\right)  (T-t)\right),
\]%
and
\[
H(t,X_t,k^-_{t,T},v_t)=\frac{e^{X_t}N^{\prime }(d_{+}\left(
k^-_{t,T},v_t\right) )}{2v_t\sqrt{T-t}}
\left( \frac{1}{v_t^2(T-t)} \left(v_t^2-I(k^-_{t,T})^2\right)  (T-t)\right).
\]%
Moreover, 
$$
(BS^{-1})'(k,u)=\frac{1}{\frac{\partial BS}{\partial\sigma}(k,BS^{-1}(k,u))}=\frac{1}{\exp(X_{t})N^{\prime }(d_+\left( k,BS^{-1}(k,u)\right))\sqrt{T-t}}.
$$
Then
\begin{eqnarray}
\lefteqn{\lim_{T\to t} \frac{\rho}{2}\frac{H(t,X_{t},k^+_{t,T},v_t)U_t(k^+_{t,T})}{(T-t)^{H+\frac12}}}
\nonumber\\
&=&\lim_{T\to t}\frac{\rho}{2}\frac{e^{X_{s}}N^{\prime }(d_{+}\left(
k^+_{t,T},v_t\right) )}{2v_t\sqrt{T-t}}
\left( \frac{1}{v_t^2(T-t)} \left(v_t^2+I(k^+_{t,T})^2\right)  (T-t)\right)\nonumber\\
&&\times\int_{t}^{T}\left( BS^{-1}\right) ^{\prime }\left( k^+_{t,T},\Gamma _{s}\right) \Phi _{s}ds\nonumber\\
&=&\frac{\rho }{2}\lim_{T\rightarrow t}\frac{1}{(T-t)^{H+\frac{3}{2}}}%
\int_{t}^{T}\frac{1}{v_{s}}\left( \sigma _{s}\int_{s}^{T}D_{s}^{W}\sigma
_{r}^{2}dr\right) ds \nonumber\\
&=&\frac{\rho }{2}\lim_{T\rightarrow t}\frac{1}{(T-t)^{H+\frac32}}%
\int_{t}^{T}\left( \int_{s}^{T}D_{s}^{W}\sigma
_{r}^{2}dr\right) ds
\end{eqnarray}
while
\begin{eqnarray}
\lim_{T\to t} \frac{\rho}{2}\frac{H(t,X_{t},k^-_{t,T},v_t)U_t(k^-_{t,T})}{(T-t)^{H+\frac12}}=0.
\end{eqnarray}
{\it Step 3} A direct computation gives us that
\begin{eqnarray*}
\lefteqn{\frac{\partial }{\partial x}\left( \frac{\partial ^{2}}{\partial x^{2}}-%
\frac{\partial }{\partial x}\right) H(t,X_{t},k^+_{t,T},v_{t})} \\
&=&\frac{e^x}{16}\frac{N^{\prime }(d_{+}\left( k^+_{t,T},v_{s}\right) ) }{ v_{t}^9(T-t)^{\frac52}}\\
&\times&
\left( (T-t)^{2}(I(k^+_{t,T})^8+2I(k^+_{t,T})^6 v_t^2
-2I(k^+_{t,T})^2 v_t^6-v_t^8)
-24(T-t)(I(k^+_{t,T})^4 v_t^2+I(k^+_{t,T})^2 v_t^4)
+48v_t^4\right),
\end{eqnarray*}
and
\begin{eqnarray*}
\lefteqn{\frac{\partial }{\partial x}\left( \frac{\partial ^{2}}{\partial x^{2}}-%
\frac{\partial }{\partial x}\right) H(t,X_{t},k^-_{t,T},v_{t})} \\
&=&\frac{e^x}{16}\frac{N^{\prime }(d_{-}\left( k^-_{t,T},v_t\right) ) }{ v_{t}^9(T-t)^{\frac52}}\\
&\times&
\left( (T-t)^{2}(I(k^-_{t,T})^8-2I(k^-_{t,T})^6 v_t^2
+2I(k^-_{t,T})^2 v_t^6-v_t^8)
+24(T-t)(I(k^-_{t,T})^4 v_t^2-I(k^-_{t,T})^2 v_t^4)
+48v_t^4\right).
\end{eqnarray*}
Notice that the leading term as $T\to t$ is the same in both expressions. Moreover, they appear in (\ref{superdecomposition}) with different sign. Then straightforward computations allow us to see that the sum of the third and the fourth terms in (\ref{superdecomposition}) is of order $O(T-t)^{2H+1}$.

{\it Step 4} Notice that
\[
\frac{\partial H}{\partial x}(t,X_{t},k^+_{t,T},v_{t})=\frac{1}{4}\frac{%
e^{X_t}N^{\prime }(d_{-}\left( k^+_{t,T},v_t\right) )}{v_t^5\left(\sqrt{T-t}%
\right) ^{3}}\left( (v_t^{2}-I(k^+_{t,T})^2)^2(T-t)-4v_t^2\right) ,
\]
and
\[
\frac{\partial H}{\partial x}(t,X_{t},k^-_{t,T},v_{t})=\frac{1}{4}\frac{%
e^{X_t}N^{\prime }(d_{-}\left( k^-_{t,T},v_t\right) )}{v_t^5\left(\sqrt{T-t}%
\right) ^{3}}\left( (v_t^{2}+I(k^-_{t,T})^2)^2(T-t)-4v_t^2\right).
\]
Notice that, as in Step 3, the leading terms are the same, and they appear in (\ref{superdecomposition}) with different sign. This allows us to see that the last two terms in (\ref{superdecomposition}) are of order $O(T-t)^{2H+1}$.

{\it Step 5} Finally, the results in Steps 1, 2, and 3, together with (\ref{superdecomposition})  allow us to complete the proof.
\end{proof}

\end{document}